\documentclass[12pt,twoside,a4paper,reqno]{amsart}
\usepackage{bulletinUPB}
\usepackage{graphics}
\usepackage{amssymb}
\usepackage{eucal}
\usepackage{amssymb}
\usepackage{amsmath}
\usepackage{tikz}
\usepackage{wrapfig}
\usepackage{courier}
\usetikzlibrary{shapes,arrows}

\tikzstyle{block} = [draw, fill=white, rectangle, 
    minimum height=3em, minimum width=6em]
\tikzstyle{sum} = [draw, fill=white, circle, node distance=1cm]
\tikzstyle{input} = [coordinate]
\tikzstyle{output} = [coordinate]
\tikzstyle{pinstyle} = [pin edge={to-,thin,black}]
\info{C}{87}{2}{2025}
\title{Properties of calculus in r-Complexity}

\author{Rares Folea}
\address{$^1$Department of Computer Science and Engineering, Faculty for Automatic Control and Computers, National University of Science and Technology Politehnica Bucharest, Romania \and
Doctoral School of Engineering and Applications of Lasers and Accelerators (S.D.I.A.L.A.), e-mail: {\tt rares.folea@stud.acs.upb.ro}}

\author{Emil Slusanschi}
\address{$^2$Department of Computer Science and Engineering, Faculty for Automatic Control and Computers, National University of Science and Technology Politehnica Bucharest, Romania, e-mail: {\tt emil.slusanschi@cs.pub.ro}}

\begin{document}
\pagestyle{headings}
\maketitle

\begin{abstract}
{This paper presents a series of general properties of the r-Complexity calculus, a complexity measurement for assessing the performance and asymptotic behaviour of real-world algorithms. This research describes characteristics such as reflexivity, transitivity, or symmetry and discusses several conversion rules between different classes of r-Complexity, as well as establishing fundamental arithmetic principles. The work also examines the behaviour of the addition property within this system and compares its characteristics with those frequently used in the traditional Bachmann-Landau notation. Through utilizing these properties, this research seeks to promote the exploration and development of novel applications for r-Complexity, as well as accelerating the adoption rate of calculus in this refined complexity model.}
\end{abstract}


\section{Introduction}
\label{intro}

Algorithms stand at the core of computer science as they are used as primary building blocks and asymptotic notations are being widely accepted as the main method for estimating its performance, by calculating the complexity of the analysed algorithm~\cite{giumale,cormen,mogos}. When examining and comparing two or more code snippets, calculating the asymptotic complexity is important for a number of reasons: it helps predicting the scalability of the solution, serves as a comparison method between different algorithms and it can help engineers understand the inherent limitations and trade-offs of each solution. Nonetheless, it can help identify bottlenecks in the code, which can then later be optimized by developers. Algorithmic complexity is commonly expressed utilizing the Bachmann-Landau~\cite{bachmann,landau} asymptotic notations, encompassing Big-Theta, Big-O, and Big-Omega measurements.

Asymptotic code complexity is a topic of interest in various fields, including logic, coding theory, and computational systems. In general, the asymptotic complexity mostly aims to map how the algorithm's utilisation of resources (e.g. total CPU time or memory) grows, with respect to increases in the input size, to later estimate how it will perform with larger inputs. Having this metric, asymptotic notations permits the comparison of multiple algorithms' efficiency, as for example an $\Theta(n \cdot log\ n)$ algorithm is more efficient than an $\Theta(n^2)$ algorithm, for sufficiently large input sizes. Over time, many asymptotic notations~\cite{knuth,giumale,cormen,wilf,bovet,mogos,affeldt,mala} have been introduced, but the fundamental concept has remained essentially the same: the definition of an asymptotic notation is based on measuring different complexity functions against a given reference function. Also, there have been efforts to try and automate the process of finding the complexity class for a given program~\cite{phalke,vaz}.

The r-Complexity model~\cite{rc} is an alternative asymptotic notation to the Bachmann-Landau notations, that offers better complexity feedback for similar programs and provides subtle insights, even for algorithms that are part of the same conventional complexity class. The model aims to address some of the shortfalls of the traditional model, such as providing a solution for making discrepancy between two similar algorithms with two similar complexity functions, $v,w:\mathbb{N}\longrightarrow\mathbb{R}$, for which that there exists $g$, such that $v(n) \in \Theta(g(n))$ and $w(n) \in \Theta(g(n))$. 

In \textbf{Section~\ref{motivation}}, we present the motivation of our research and illustrate a scenario where two algorithms solving the same problem would have been grouped into the same complexity class in the traditional analysis, despite their vastly different performance in practice. Then, in \textbf{Section~\ref{definitions}}, we present some formal definitions and associated corollaries for the various r-Complexity classes, including Big r-Theta, Big r-O and Big r-Omega, as well as presenting alternative criteria for admittance in a given set. \textbf{Section~\ref{properties}} explores properties like reflexivity, transitivity, symmetry, and projections within the context of r-Complexity. \textbf{Section~\ref{addition}} examines the behaviour of r-Complexity classes under addition and finally, in \textbf{Section~\ref{conclusion}}, we summarize the key results of the research.

\section{Motivation}
\label{motivation}

The r-Complexity model demonstrates value by analysing dominant constants, revealing significant differences in some algorithms' behaviour, that can substantially impact the total execution time. To help software engineers collaborate and share knowledge more directly, by discussing and comparing algorithms using the r-Complexity asymptotic notation, this paper presents a set of common properties of calculus within the r-Complexity classes, that can be used to obtain a faster pace of operating with complexities. 

In the traditional Bachmann-Landau notations, the Big-Theta definition~\cite{knuth,giumale,cormen} states that a function $f(n)$ belongs to the $\Theta(g(n))$ complexity set if and only if two positive constants, denoted by $c_1$ and $c_2$, exists, as well as a positive integer $n_0$, such that $f(n)$ is bounded between $g(n)$, multiplied by $c_1$ (lower bound) and $c_2$ (upper bound) constants, for any value of $n$ greater than or equal to $n_0$:

\[\Theta(g(n)) = \{ f(n) \mid \exists c_1, c_2 \in \mathbb{R}^{*}_{+}, \exists n_0 \in \mathbb{N}^{*} \ s.t. \  c_1 \cdot g(n) \leq f(n) \leq c_2 \cdot g(n), \ \forall n \geq n_0 \}\]

To highlight potential situations where the lack of distinction can mislead developers, consider two algorithms, $Alg1$ and $Alg2$, that solve the same problem are defined by different complexity functions: $f_{1},f_{2}:\mathbb{N}\longrightarrow\mathbb{R}$, with the properties that the functions defer exactly by one constant: $f_{1} = x \cdot f_{2}$, with $x \in \mathbb{R}_{+}, \ x > 1$. 

These complexities functions will imply that the two algorithms would belong to the same Big-Theta complexity set, because there exists $g$, with $f_{2} \in \Theta(g(n))$, such as:

\[ \exists c_{1}, c_{2} \in \mathbb{R}^{*}_{+}, \exists n_{0} \in \mathbb{N}^{*}\ s.t.\ \ c_{1} \cdot g(n) \leq f_{2}(n) \leq c_{2} \cdot g(n)\ ,\  \forall n \geq n_{0} \]

holds for $ c_{1}^{'}, c_{2}^{'} \in \mathbb{R}^{*}_{+}, n_{0}^{'} \in \mathbb{N}^{*}$, where: 

$
\left\{
\begin{aligned}
c_{1}^{'} &= x \cdot c_{1} \\
c_{2}^{'} &= x \cdot c_{2} \\
n_{0}^{'} &= n_{0} 
\end{aligned}
\right.
$

because $\forall n \geq n_{0}^{'}$: \[c_{1}^{'} \cdot g(n) = x \cdot c_{1} \cdot g(n) \leq x\cdot f_{2}(n) = f_{1}(n) = x\cdot f_{2}(n) \leq x \cdot c_{2} \cdot g(n) = c_{2}^{'} \cdot g(n)\]

This implies:

\[
c_{1}^{'} \cdot g(n) \leq f_{1}(n) \leq c_{2}^{'} \cdot g(n)
\]

and therefore, based to the definition of Big Theta, $f_{2} \in \Theta(g(n)) \Rightarrow f_{1} \in \Theta(g(n))$. 

Similarly, $f_{1} \in \Theta(g(n)) \Rightarrow f_{2} \in \Theta(g(n))$ for $c_{1}^{'} = \dfrac{1}{x} \cdot c_{1}$, $c_{2}^{'} = \dfrac{1}{x} \cdot c_{2}$ and $n_{0}^{'} = n_{0}$.

Remark that even if $f_{1} > f_{2}$, both complexity functions are part of the same complexity class. This observation implies that two algorithms, whose complexity functions can differ by a constant, even as high as $2024^{2024}$, are part of the same complexity class, even if the actual run-time might differ by over five-thousands orders of magnitude. For comparison, only a $30$ magnitude order between the two complexity functions $f_{1} = 10^{30} \cdot f_{2}$, signify
that if for a given input $n$, if $Alg2$ ends execution in 1 \textit{attosecond} ($10^{-9}$ part of a nanosecond), then $Alg1$ is expected to end execution in about 3 \textit{millenniums}. Despite of the colossal difference in time, classical complexity model is not perceptive between these subset of algorithms, that have their complexity functions differing by only some constants. In general, it is easier to evaluate the algorithm's effectiveness and applicability for various scenarios when one is aware of these complexity features.

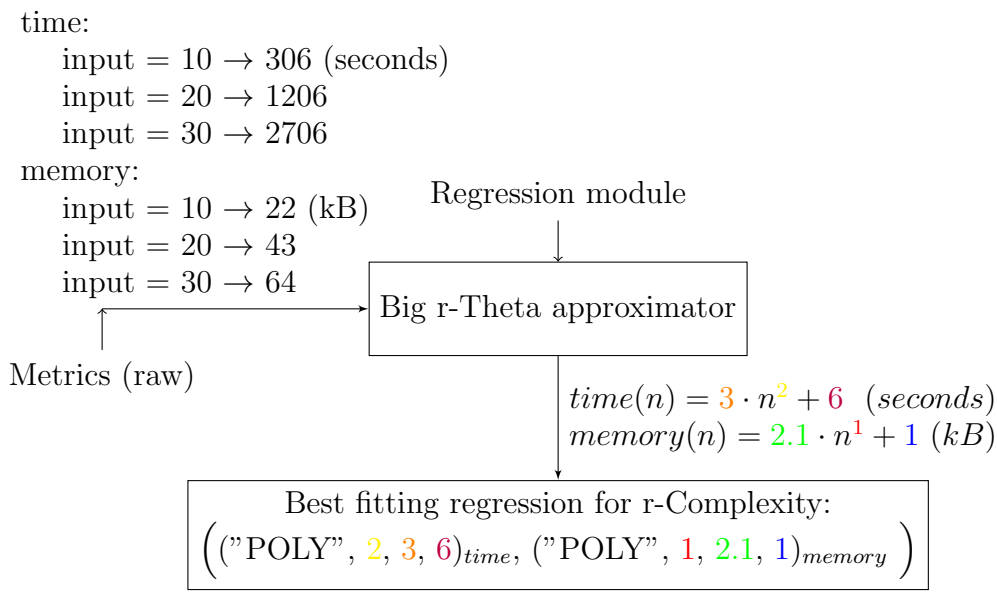
\begin{figure}
\centering
\begin{tikzpicture}[auto, node distance=3cm,>=latex']
    \node [input, name=input, pin={[pinstyle]below:Metrics (raw)}] {};
    \node [block, right of=input, pin={[pinstyle]above:Regression module},  node distance=6cm, align=center] (system) {Big r-Theta approximator};
    \node [block, below of=system,  node distance=3cm, align=center] (birthmark) {Best fitting regression for r-Complexity: \\
    $\Bigl($("POLY", \textcolor{yellow}{2}, \textcolor{orange}{3}, \textcolor{purple}{6})$_{time}$, ("POLY", \textcolor{red}{1}, \textcolor{green}{2.1}, \textcolor{blue}{1})$_{memory}$ $\Bigl)$
    };

    \draw [->] (input) -- node [name=p, align=left] { 
    time:\\
    \ \ \ \ input = 10 $\rightarrow$ 306 (seconds) \\
    \ \ \ \ input = 20 $\rightarrow$ 1206 \\
    \ \ \ \ input = 30 $\rightarrow$ 2706 \\
    memory:\\
    \ \ \ \ input = 10 $\rightarrow$ 22 (kB) \\
    \ \ \ \ input = 20 $\rightarrow$ 43 \\
    \ \ \ \ input = 30 $\rightarrow$ 64 }(system);

    \draw [->] (system) -- node [name=p, align=center] {
    $time(n) = \textcolor{orange}{3} \cdot n^{\textcolor{yellow}{2}} + \textcolor{purple}{6}\ \  (seconds)$ \\
    $memory(n) = \textcolor{green}{2.1} \cdot n^{\textcolor{red}{1}} + \textcolor{blue}{1}\ (kB)$ 
    } (birthmark); 

\end{tikzpicture}

\caption{An example of the regression fitting process in~\cite{cbce}, with a set of two raw metrics: temporarily and space. The goal of the research is to approximate, using simplified functions, the actual Big r-Theta complexity function that models program behaviour across multiple performance metrics, which later can be used to build complexity-based code embeddings. One part of the embedding that analyses the time complexity, \textsf{POLY 2 3 6}, denotes a second degree polynomial, with the coefficient $3$ for the squared term of the polynomial and the intercept $6$.}
\label{fig:regression-module}
\end{figure}

Some applications of r-Complexity have been presented in~\cite{rc}, which has challenged the common belief that Strassen's algorithm, with its lower asymptotic complexity would always outperform a traditional matrix multiplication for any input size. The result in the paper has shown that for input sizes less than $25$ million element, a cache-friendly loop ordering algorithm would finish faster due, despite the increase in complexity. The same paper presents an estimation of the total time using a brute-force algorithm for generating all possible episodes to around $10^{123}$ seconds, based on the associated r-Complexity function. This has been the core argument for arguing that solving the game of chess perfectly, by brute force algorithms, is currently infeasible, due to the limitations of available computing power. This argument could have not been established by using the traditional complexity models, for finite inputs.

Another area where r-Complexity has proven useful~\cite{cbce} is in building code-embeddings, representative of the program logic, that would later be used to classify algorithms into different categories, such as greedy, brute force or divide and conquer. The embedding calculus is based on the idea of finding the coefficients of the associated r-Complexity function, estimated by fitting a regression model to the metrics obtained from the profilers during the execution of algorithm, as a function of its input size. For example, a real use-case presented in Figure~\ref{fig:regression-module}, indicates, via experimental data, that the algorithm's execution time is proportional to the square of the input size and the amount of memory used by the algorithm grows proportionally to the size of the input data. Using r-Complexity models for estimating the Big r-Theta class, we can have insights into what the proportion is, such as insights that the execution time is scaling three times with the square of the input size, or that the memory scales $2.1$ times as fast as the input size.

\section{Definition of the r-Complexity sets}
\label{definitions}

The definition of the r-Complexity sets have initially been introduced in~\cite{rc}. In this paper, we extend the expressively of these statements by associating corollaries, build upon the knowledge described in the definition, which will highlight important consequences in asymptotic calculus. The validity of the corollaries presented in this section follows directly from the definition of the classes.

Calculus in r-Complexity can be performed either using limits of sequences or limits of functions. Consider any two complexity functions $f,g:\mathbb{N}_{+}\longrightarrow\mathbb{R}$. The following notations and terminology will be used to describe the asymptotic behaviour of an algorithm's complexity, as characterized by a function $f:\mathbb{N}\longrightarrow\mathbb{R}$. We define the set of all complexity calculus $\mathcal{F}= \lbrace f:\mathbb{N}\longrightarrow\mathbb{R} \rbrace$.

\begin{definition}
    Big \textit{r-}Theta has been defined as the set that represents the group of mathematical functions that have the same growth rate as a function $g$, in the context of asymptotic analysis. More formally, the Big \textit{r-}Theta r-Complexity class has been defined~\cite{rc} as:
    \[\begin{split}
          \Theta_{r}(g(n)) = \lbrace f \in \mathcal{F}\ |\ \forall c_{1}, c_{2} \in \mathbb{R}^{*}_{+} \ s.t. c_{1}< r < c_{2} , \exists n_{0} \in \mathbb{N}^{*}\ \\ s.t.\ \ c_{1} \cdot g(n) \leq f(n) \leq c_{2} \cdot g(n)\ ,\  \forall n \geq n_{0} \rbrace
    \end{split} \]
\end{definition}

\begin{corollary}
    Criteria for a function $f$ to be in the Big \textit{r-}Theta complexity class defined by a function $g$ is a direct result of the following:
    \[ f \in \Theta_{r}(g(n)) \Leftrightarrow \lim_{n\to\infty} \dfrac{f(n)}{g(n)} = r \]
\end{corollary}

We will use the same definition for the Big \textit{r-}O, Big \textit{r-}Omega, Small \textit{r-}O and Small \textit{r-}Omega as presented in~\cite{rc}, with the remark that the last two models are defined for completeness reasons only, as they don't provide any novelty opposed to the traditional notations~\cite{papadimitriou,wegener}. Similar to Big \textit{r-}Theta, following the definitions of Big \textit{r-}O, Big \textit{r-}Omega, Small \textit{r-}O and Small \textit{r-}Omega classes, the following corollaries are direct implications:

\begin{corollary}
    Criteria for admittance of $f$ in Big \textit{r-}O class defined by $g$:
    \[ f \in \mathcal{O}_{r}(g(n)) \Leftrightarrow \lim_{n\to\infty} \dfrac{f(n)}{g(n)} = l,\ l \in \left[ 0, r \right] \]
\end{corollary}

\begin{corollary}
    Criteria for admittance of $f$ in Big \textit{r-}Omega class defined by $g$:
    \[ f \in \Omega_{r}(g(n)) \Leftrightarrow \lim_{n\to\infty} \dfrac{f(n)}{g(n)} = l,\ l \in \left[ r, \infty \right) \]
\end{corollary}

The criteria for admittance of $f$ in both Small \textit{r-}O and Small \textit{r-}Omega complexity classes are unchanged from the traditional notations.

\begin{corollary}
    Criteria for admittance of $f$ in Small \textit{r-}O class defined by $g$:
    \[ f \in o_{r}(g(n)) \Leftrightarrow \lim_{n\to\infty} \dfrac{f(n)}{g(n)} = 0 \]
\end{corollary}

\begin{corollary}
    Criteria for admittance of $f$ in Small \textit{r-}Omega class defined by $g$:
    \[ f \in \omega_{r}(g(n)) \Leftrightarrow \lim_{n\to\infty} \dfrac{f(n)}{g(n)} = \infty \]
\end{corollary}

\section{Common properties in r-Complexity}
\label{properties}

This chapter will present new implications in terms of \textit{reflexivity, transitivity, symmetry and projections} in r-Complexity, and it will draw the comparison with equivalent properties within the conventional Bachmann–Landau notations, already well analysed and established in the literature ~\cite{mogos,affeldt,mala,goloveshkin}. By recognizing reflexivity, symmetry and transpose properties in r-Complexity, we consider that some problems can be solved more efficiently in this model. To establish the results in this section, we will rely on the definitions introduced in Section~\ref{definitions}, and the foundational model described in~\cite{rc}.

\begin{theorem}
    Reflexivity in r-Complexity for Big r-Theta notation  \[ f \in \Theta_{r} \left( \frac{1}{r} \cdot f(n) \right)\ \forall r \neq 0 \]
\end{theorem}
\begin{proof}
    Based on the definition of $ \Theta_{r}(f(n))$
    \[\begin{split}
          \Theta_{r}(f(n)) = \lbrace f' \in \mathcal{F}\ |\ \forall c_{1}, c_{2} \in \mathbb{R}^{*}_{+} \ s.t. c_{1}< r < c_{2} , \exists n_{0} \in \mathbb{N}^{*}\ \\ s.t.\ \ c_{1} \cdot f(n) \leq f'(n) \leq c_{2} \cdot f(n)\ ,\  \forall n \geq n_{0} \rbrace
    \end{split} \]
    Using substitution $ f(n) \longleftarrow \frac{1}{r} \cdot f(n)$
    \[\begin{split}
          \Theta_{r} \left( \frac{1}{r} \cdot f(n) \right) = \lbrace f' \in \mathcal{F}\ |\ \forall c_{1}, c_{2} \in \mathbb{R}^{*}_{+} \ s.t. c_{1}< r < c_{2} , \exists n_{0} \in \mathbb{N}^{*}\ \\ s.t.\ \frac{1}{r} \cdot \ c_{1} \cdot f(n) \leq  f'(n) \leq \frac{1}{r} \cdot c_{2} \cdot f(n)\ ,\  \forall n \geq n_{0} \rbrace
    \end{split} \]
    By choosing $c_{1}' = \frac{1}{r} \cdot c_{1}, c_{2}' = \frac{1}{r} \cdot c_{2}$
    \[\begin{split}
          \Theta_{r} \left( \frac{1}{r} \cdot f(n) \right) = \lbrace f' \in \mathcal{F}\ |\ \forall c_{1}', c_{2}' \in \mathbb{R}^{*}_{+} \ s.t. c_{1}'< 1 < c_{2}' , \exists n_{0} \in \mathbb{N}^{*}\ \\ s.t.\  \ c_{1}' \cdot f(n) \leq  f'(n) \leq c_{2}' \cdot f(n)\ ,\  \forall n \geq n_{0} \rbrace
    \end{split} \]
    For any $x \in \mathbb{R}^{*}$,  $\forall c_{1}, c_{2}\ \ s.t. c_{1} \leq 1 \leq c_{2}$, we have $ x \cdot c_{1} \leq x \leq x \cdot c_{2} $.
    Thus, if $f:\mathbb{N}\longrightarrow\mathbb{R}$, $\forall c_{1}, c_{2}\ \  s.t. c_{1} \leq 1 \leq c_{2}$, we have $f(n) \cdot c_{1} \leq f(n) \leq f(n) \cdot c_{2}\ \ \forall n \in \mathbb{N}^{*}$ or $f(n) \cdot c_{1} \geq f(n) \geq f(n) \cdot c_{2}\ \ \forall n \in \mathbb{N}^{*}$

    Therefore:
    \[\forall c_{1}', c_{2}' \in \mathbb{R}^{*}_{+} \ s.t. c_{1}'< 1 < c_{2}' , \exists n_{0} = 1 s.t.\  \ c_{1}' \cdot f(n) \leq  f(n) \leq c_{2}' \cdot f(n)\ ,\  \forall n \geq n_{0}=1 \Rightarrow \]
    \[ f \in \Theta_{r} \left( \frac{1}{r} \cdot f(n) \right)\ \forall r \neq 0 \]
\end{proof}

\begin{theorem}
    Reflexivity in r-Complexity for Big r-O notation
    \[ f \in \mathcal{O}_{r} \left( x \cdot f(n) \right)\ \forall x \geq \dfrac{1}{r} \]
\end{theorem}
\begin{proof}
    Based on the definition of $ \mathcal{O}_{r}(f(n))$
    \[\mathcal{O}_{r}(f(n)) = \lbrace f' \in \mathcal{F}\ |\ \forall c  \in \mathbb{R}^{*}_{+} \ s.t.\  r<c, \exists n_{0} \in \mathbb{N}^{*}\ s.t.\  f'(n) \leq c \cdot f(n),\  \forall n \geq n_{0} \rbrace\]
    Using substitution $ f(n) \longleftarrow x \cdot f(n), \forall x \geq \dfrac{1}{r}, \forall r \neq 0$
    \[\mathcal{O}_{r}(x \cdot f(n)) = \lbrace f' \in \mathcal{F}\ |\ \forall c  \in \mathbb{R}^{*}_{+} \ s.t.\  r<c, \exists n_{0} \in \mathbb{N}^{*}\ s.t.\  f'(n) \leq c \cdot x \cdot f(n),\  \forall n \geq n_{0} \rbrace\]
    If $r<c$ and $x \geq \dfrac{1}{r}$, then $c \cdot x \geq 1, \forall r,c,x\ \ r \neq 0$
    Thus, if $f:\mathbb{N}\longrightarrow\mathbb{R}$, we have $ f(n) \leq 1 \cdot f(n) \leq  c \cdot x\cdot f(n) \ \ \forall n \in \mathbb{N}^{*}$.

    Therefore:
    \[\forall x \geq \dfrac{1}{r}, \ \forall c \in \mathbb{R}^{*}_{+} \ s.t.\  r<c , \exists n_{0} = 1 s.t.\  \  f(n) \leq  c \cdot x\cdot f(n) \ \ \forall n \in \mathbb{N}^{*} ,\  \forall n \geq n_{0}=1 \Rightarrow \]
    \[ f \in \mathcal{O}_{r} \left( x \cdot f(n) \right)\ \forall x \geq \dfrac{1}{r} \]
\end{proof}

\begin{theorem}
    Reflexivity in r-Complexity for Big r-Omega notation
    \[ f \in \Omega_{r} \left( x \cdot f(n) \right)\ \forall x \leq \dfrac{1}{r} \]
\end{theorem}
\begin{proof}
    Based on the definition of $\Omega_{r}(f(n))$
    \[\Omega_{r}(f(n)) = \lbrace f' \in \mathcal{F}\ |\ \forall c  \in \mathbb{R}^{*}_{+} \ s.t.\  c<r, \exists n_{0} \in \mathbb{N}^{*}\ s.t.\  f'(n) \geq c \cdot f(n),\  \forall n \geq n_{0} \rbrace\]
    Using substitution $ f(n) \longleftarrow x \cdot f(n), \forall x \leq \dfrac{1}{r}, \forall r \neq 0$
    \[\Omega_{r}(x \cdot f(n)) = \lbrace f' \in \mathcal{F}\ |\ \forall c  \in \mathbb{R}^{*}_{+} \ s.t.\  c<r, \exists n_{0} \in \mathbb{N}^{*}\ s.t.\  f'(n) \geq x \cdot c \cdot f(n),\  \forall n \geq n_{0} \rbrace\]
    If $r>c$ and $x \leq \dfrac{1}{r}$, then $c \cdot x \leq 1, \forall r,c,x\ \ r \neq 0$

    Therefore:
    \[\forall x \leq \dfrac{1}{r}, \ \forall c \in \mathbb{R}^{*}_{+} \ s.t.\  c<r , \exists n_{0} = 1 s.t.\  \  f(n) \geq  c \cdot x\cdot f(n) \ \ \forall n \in \mathbb{N}^{*} ,\  \forall n \geq n_{0}=1 \Rightarrow \]
    \[ f \in \Omega_{r} \left( x \cdot f(n) \right)\ \forall x \leq \dfrac{1}{r} \]
\end{proof}

\begin{remark}
    Reflexivity does not hold in r-Complexity for small r-O and small r-Omega notation, as this set is the same as the classical sets defined in Bachmann–Landau notations, and reflexivity does not hold for Small Omega class.
    \[ f \notin o_{r}(f(n)) \]
    \[ f \notin \omega_{r}(f(n)) \]
\end{remark}

\begin{theorem}
    Transitivity in r-Complexity - Big r-Theta notation:
    
    \[ f \in \Theta_{r}(g(n)), g \in \Theta_{r'}(h(n)) \Rightarrow  f \in \Theta_{r \cdot r'}(h(n))\]
\end{theorem}
\begin{proof}

$ f \in \Theta_{r}(g(n)) \Rightarrow \forall c_{1}, c_{2} \in \mathbb{R}^{*}_{+} \ s.t. c_{1}< r < c_{2} , \exists n_{0} \in \mathbb{N}^{*}\ \\ s.t.\ \ c_{1} \cdot g(n) \leq f(n) \leq c_{2} \cdot g(n)\ ,\  \forall n \geq n_{0} $ 
    
$ g \in \Theta_{r'}(h(n)) \Rightarrow \forall c_{1}, c_{2} \in \mathbb{R}^{*}_{+} \ s.t. c_{1} < r' < c_{2} , \exists n_{0}' \in \mathbb{N}^{*}\ \\ s.t.\ \ c_{1} \cdot h(n) \leq g(n) \leq c_{2} \cdot h(n)\ ,\  \forall n \geq n_{0}' $ 
    
Therefore: \[
\forall c_{1}, c_{2}, c_{1}', c_{2}' \in \mathbb{R}^{*}_{+} \ s.t. c_{1} < r < c_{2}, c_{1}' < r' < c_{2}' , \exists n''_{0}=max(n_{0}, n'_{0}) \in \mathbb{N}^{*}\ s.t.
\] 
\[c_{1}' \cdot c_{1} \cdot h(n) \leq c_{1} \cdot g(n) \leq f(n) \leq c_{2} \cdot g(n) \leq c_{2}' \cdot c_{2} \cdot h(n)\ ,\  \forall n \geq n''_{0} 
\]

For $c_{1}'' = c_{1} \cdot c_{1}' , c_{2}'' = c_{2} \cdot c_{2}'$, we have: $c_{1}'' < r \cdot r' < c_{2}''$.

As a consequence:

\[\forall c_{1}'', c_{2}'' \in \mathbb{R}^{*}_{+} \ s.t. c_{1}'' < r \cdot r' < c_{2}'' \ \ \exists n''_{0}=max(n_{0}, n'_{0}) \in \mathbb{N}^{*}\ \\ s.t.\] \[c_{1}'' \cdot h(n) \leq f(n) \leq c_{2}'' \cdot h(n)\ ,\  \forall n \geq n''_{0} \Rightarrow f \in \Theta_{r \cdot r'}(h(n))\]

\end{proof}

The remaining transitivity properties also hold within r-Complexity Calculus, and their proofs follow a similar structure to the one presented above:

\begin{theorem}
    Transitivity in r-Complexity - Big r-O notation:

    $ f \in \mathcal{O}_{r}(g(n)), g \in \mathcal{O}_{r}(h(n)) \Rightarrow  f \in \mathcal{O}_{r \cdot r'}(h(n))$
\end{theorem}
\begin{theorem}
    Transitivity in r-Complexity - Big r-Omega notation:

    $ f \in \Omega_{r}(g(n)), g \in \Omega_{r}(h(n)) \Rightarrow  f \in \Omega_{r \cdot r'}(h(n))$
\end{theorem}
\begin{theorem}
    Transitivity in r-Complexity - Small r-O notation:

    $ f \in o_{r}(g(n)), g \in o_{r'}(h(n)) \Rightarrow  f \in o_{r \cdot r'}(h(n))$
\end{theorem}
\begin{theorem}
    Transitivity in r-Complexity - Small r-Omega notation

    $ f \in \omega_{r}(g(n)), g \in \omega_{r'}(h(n)) \Rightarrow  f \in \omega_{r \cdot r'}(h(n))$
\end{theorem}

Symmetry is a fundamental algebraic property that plays a crucial role in various mathematical contexts, as it simplifies problem solving in many cases, and it also provides means to perform a more elegant flow of calculations.

\begin{theorem}
    Symmetry in r-Complexity:  $ f \in \Theta_{r}(g(n)) \Rightarrow g \in \Theta_{\frac{1}{r}}(f(n)) $
\end{theorem}

\begin{proof}
    $ f \in \Theta_{r}(g(n)) \Rightarrow \forall c_{1}, c_{2} \in \mathbb{R}^{*}_{+} \ s.t. c_{1}< r < c_{2} , \exists n_{0} \in \mathbb{N}^{*}\ \\ s.t.\ \ c_{1} \cdot g(n) \leq f(n) \leq c_{2} \cdot g(n)\ ,\  \forall n \geq n_{0} $
    \\ Using the substitution $c_{1}' = \dfrac{c_{1}}{r}, c_{2}' = \dfrac{c_{2}}{r} \Rightarrow \forall c_{1}', c_{2}' \in \mathbb{R}^{*}_{+} \ s.t. c_{1} < 1 < c_{2} , \exists n_{0} \in \mathbb{N}^{*}\ \\ s.t.\ \ c_{1}' \cdot g(n) \leq \dfrac{1}{r} \cdot f(n) \leq c_{2}' \cdot g(n)\ ,\  \forall n \geq n_{0} $
    \\ The previous inequality can be re-written as:
    \[\begin{cases}
          g(n) \leq \dfrac{1}{c_{1}'} \cdot \dfrac{1}{r} \cdot f(n) \\ g(n) \geq \dfrac{1}{c_{2}'} \cdot \dfrac{1}{r} \cdot f(n)
    \end{cases}\]
    \\ Using notation $c_{1}'' = \dfrac{1}{c_{2}'}, c_{2}'' = \dfrac{1}{c_{1}'}$ the inequality becomes:  \\
    $\forall c_{1},'' c_{2}'' \in \mathbb{R}^{*}_{+} \ s.t. c_{1}'' < r < c_{2}'' , \exists n_{0} \in \mathbb{N}^{*}\ $
    \[ {c_{1}''} \cdot \dfrac{1}{r} \cdot f(n) \leq g(n) \leq {c_{2}'} \cdot \dfrac{1}{r} \cdot f(n)\ \ \forall n \geq n_{0} \]
    Thus, Based on the definition of $ \Theta_{r}(f(n))$, $ f \in \Theta_{r}(g(n)) \Rightarrow g \in \Theta_{\frac{1}{r}}(f(n)) $.
\end{proof}

\begin{theorem}
    Transpose symmetry in r-Complexity:  $ f \in \mathcal{O}_{r}(g(n)) \Leftrightarrow g \in \Omega_{\frac{1}{r}}(f(n)) $
\end{theorem}

\begin{proof}
    Using the definition of $ \mathcal{O}_{r}(f(n))$ and $f \in \mathcal{O}_{r}(g(n))$ ($\mathcal{O}_{r}(g(n)) \supseteq \lbrace f \rbrace \ $):
    \[ \ \forall c  \in \mathbb{R}^{*}_{+} \ s.t.\  r<c, \exists n_{0} \in \mathbb{N}^{*}\ s.t.\  g(n) \leq c \cdot f(n),\  \forall n \geq n_{0} \]
    
    We can rewrite the previous relation as following:
    \[ \ \forall c  \in \mathbb{R}^{*}_{+} \ s.t.\  r<c, \exists n_{0} \in \mathbb{N}^{*}\ s.t.\  \dfrac{1}{c} \cdot g(n) \leq f(n),\  \forall n \geq n_{0} \]
    Substituting the constant $ c' = \dfrac{1}{c}$:
    \[ \ \forall c'  \in \mathbb{R}^{*}_{+} \ s.t.\  c' < \dfrac{1}{r}, \exists n_{0} \in \mathbb{N}^{*}\ s.t.\  f(n) \geq c' \cdot g(n) ,\  \forall n \geq n_{0} \]

    Therefore:
    \[\Omega_{\frac{1}{r}}(f(n)) \supseteq \lbrace g \rbrace \]
\end{proof}
\begin{remark}
    The Transpose symmetry hold for Small \textit{r-}o and Small \textit{r-}Omega, as these two sets are equal with the classical sets defined in Bachmann–Landau notations.
\end{remark}

\begin{theorem}
    Projection property in r-Complexity:  \\  $ f \in \Theta_{r}(g(n)) \Leftrightarrow f \in \mathcal{O}_{r}(g(n))$ and $f \in \Omega_{r}(g(n)) $
\end{theorem}
\begin{proof}
    The proof is straightforward, using the definitions of Big \textit{r-}Theta, Big \textit{r-}O and Big \textit{r-}Omega. \\
    " $\Rightarrow$ " $ f \in \Theta_{r}(g(n)) \Rightarrow f \in \mathcal{O}_{r}(g(n)), f \in \Omega_{r}(g(n)) $
    Based on the definition of $\Theta_{r}(g(n))$, $ f \in \Theta_{r}(g(n)) \Rightarrow$
    \[\forall c_{1}, c_{2} \in \mathbb{R}^{*}_{+} \ s.t. c_{1} < r < c_{2} , \exists n_{0} \in \mathbb{N}^{*}\ \\ s.t.\ \ c_{1} \cdot g(n) \leq f(n) \leq c_{2} \cdot g(n)\ ,\  \forall n \geq n_{0} \]
    By splitting the inequality:
    \[\begin{cases}
          \forall c_{1} \in \mathbb{R}^{*}_{+} \ s.t. c_{1} < r , \exists n_{0} \in \mathbb{N}^{*}\ s.t.\ \ c_{1} \cdot g(n) \leq f(n) ,\  \forall n \geq n_{0} \\ \forall  c_{2} \in \mathbb{R}^{*}_{+} \ s.t. r < c_{2} , \exists n_{0} \in \mathbb{N}^{*}\ s.t.\ f(n) \leq c_{2} \cdot g(n)\ ,\  \forall n \geq n_{0}
    \end{cases}\]
    Therefore:
    \[\begin{cases}
          f \in \mathcal{O}_{r}(g(n)) \\ f \in \Omega_{r}(g(n))
    \end{cases}\]

    " $\Leftarrow$ " $ f \in \mathcal{O}_{r}(g(n)), f \in \Omega_{r}(g(n)) \Rightarrow f \in \Theta_{r}(g(n)) $
    Using the definitions of Big \textit{r-}O and Big \textit{r-}Omega:
    \[\begin{cases}
          \forall c_{1} \in \mathbb{R}^{*}_{+} \ s.t. c_{1} < r , \exists n_{0} \in \mathbb{N}^{*}\ s.t.\ \ c_{1} \cdot g(n) \leq f(n) ,\  \forall n \geq n_{0} \\ \forall  c_{2} \in \mathbb{R}^{*}_{+} \ s.t. r < c_{2} , \exists n_{0}' \in \mathbb{N}^{*}\ s.t.\ f(n) \leq c_{2} \cdot g(n)\ ,\  \forall n \geq n_{0}'
    \end{cases}\]
    By choosing $n_{0}'' = max(n_{0}, n_{0}')$:
    \[\begin{cases}
          \forall c_{1} \in \mathbb{R}^{*}_{+} \ s.t. c_{1} < r ,\Rightarrow\ \ c_{1} \cdot g(n) \leq f(n) ,\  \forall n \geq n_{0}'' \\ \forall  c_{2} \in \mathbb{R}^{*}_{+} \ s.t. r < c_{2} , \Rightarrow \ f(n) \leq c_{2} \cdot g(n)\ ,\  \forall n \geq n_{0}''
    \end{cases}\]
    By merging the inequalities, we have for $\forall n \geq n_{0}''$:
    \[\forall c_{1}, c_{2} \in \mathbb{R}^{*}_{+} \ s.t. c_{1} < r < c_{2} , \exists n_{0}'' = max(n_{0}, n_{0}')\  s.t.\ \ c_{1} \cdot g(n) \leq f(n) \leq c_{2} \cdot g(n)\ \]
    Therefore $f \in \Theta_{r}(g(n))$.
\end{proof}

\section{Addition properties}
\label{addition}

Analogous to the calculus in Bachmann-Landau notations (including Big-O arithmetic), valuable insights can be gained by examining the behaviour of r-Complexity classes under addition. For each of the presented theorem, we present a corollary, using a \textit{relax notation}, where consider that by any r-Complexity class notation, we denote an arbitrary function part of the set. This would allow us to write equations such as the one presented below.

For analysing the addition properties in Big r-Theta, the following relations hold for any correctly defined functions $f, g, f', g', h:\mathbb{N}\longrightarrow\mathbb{R}$, where \[ h(n) = f'(n) + g'(n)\ \] $ \forall n \in  \mathbb{N} $, where $f',g'$ are two arbitrary functions such that $ f' \in \Theta_{r}(f), g' \in \Theta_{q}(g) $ and $r,q \in \mathbb{R}_{+}$:
    \begin{theorem}
        If $ \lim_{n\to\infty} \dfrac{f(n)}{g(n)} = 0 \Rightarrow  h \in \Theta_{q}(g) $. \\
    \end{theorem}
    \begin{proof}
        $ \lim_{n\to\infty} \dfrac{f(n)}{g(n)} = 0 \Rightarrow \lim_{n\to\infty} \dfrac{f'(n)}{g'(n)} = 0 \Rightarrow \lim_{n\to\infty} \dfrac{g'(n) + f'(n)}{g'(n)} = 1 $ and using the result from asymptotic analysis section, we have $ g'(n) + f'(n) = h(n) \in \Theta_{1}(g')$. \\ Using reflexivity property, $ h(n) \in \Theta_{q}(g)$.
    \end{proof}
    \begin{corollary} If $\lim_{n\to\infty} \dfrac{f(n)}{g(n)} = 0$, the following relation holds:
        \[  \lim_{n\to\infty} \dfrac{f(n)}{g(n)} = 0 \Rightarrow \Theta_{r}(f) + \Theta_{q}(g) = \Theta_{q}(g)\]
    \end{corollary}

    \begin{theorem}
        If $ \lim_{n\to\infty} \dfrac{f(n)}{g(n)} = \infty \Rightarrow  h \in \Theta_{r}(f) $. \\
    \end{theorem}
    \begin{proof}
        Using the last result, by performing the swap:$f \leftarrow g, g \leftarrow f$ and considering the commutativity of addition, we obtain the proof for this statement.
    \end{proof}

    \begin{corollary} If $\lim_{n\to\infty} \dfrac{f(n)}{g(n)} = \infty$, the following relation holds:
        \[  \lim_{n\to\infty} \dfrac{f(n)}{g(n)} = \infty \Rightarrow \Theta_{r}(f) + \Theta_{q}(g) = \Theta_{r}(f)\]
    \end{corollary}

    \begin{theorem}
        If $ \lim_{n\to\infty} \dfrac{f(n)}{g(n)} = t, \ t \in \mathbb{R}_{+} \Rightarrow  h \in \Theta_{r} \left( f + \dfrac{r}{q} \cdot g \right) $. \\
    \end{theorem}
    \begin{proof}

        $ \lim_{n\to\infty} \dfrac{f(n)}{g(n)} = t \Rightarrow \lim_{n\to\infty} \dfrac{f'(n)}{g'(n)} = t \cdot \dfrac{r}{q} = t' \Rightarrow \lim_{n\to\infty} \dfrac{g'(n) + f'(n)}{g'(n)} = t' + 1 $ and using the result from asymptotic analysis section, we have $ g'(n) + f'(n) = h(n) \in \Theta_{t' + 1}(g')$. \\
        Using reflexivity property, $ h(n) \in \Theta_{t' + 1} (q \cdot g)$. \\
        Using the conversion technique, we have $ h(n) \in \Theta_{r}( \dfrac{1}{r} \cdot t \cdot \dfrac{r}{q} + 1) \cdot q \cdot g)$. \\
        By swapping back $t$, asymptotically, we can establish: 
        
        $ h(n) \in \Theta_{r} \left( \dfrac{1}{r} \cdot ( \dfrac{f(n)}{g(n)} \cdot \dfrac{r}{q} + 1) \cdot q \cdot g \right) $ \\
        Therefore $ h \in \Theta_{r} \left( f + \dfrac{r}{q} \cdot g \right) $.
    \end{proof}
    \begin{corollary} If $\lim_{n\to\infty} \dfrac{f(n)}{g(n)} = t, \ t \in \mathbb{R}_{+}$, the following relation holds:
        \[ \lim_{n\to\infty} \dfrac{f(n)}{g(n)} = t, \ t \in \mathbb{R}_{+} \Rightarrow \Theta_{r}(f) + \Theta_{q}(g) = \Theta_{r}(f + \dfrac{r}{q} \cdot g)\]
    \end{corollary}

For addition in Big r-O, the following relations hold for any correctly defined functions $f, g, f', g', h:\mathbb{N}\longrightarrow\mathbb{R}$, where $ h(n) = f'(n) + g'(n)\  \forall n \in \mathbb{N} $, where $f',g'$ are two arbitrary functions such that $ f' \in \mathcal{O}_{r}(f), g' \in \mathcal{O}_{q}(g) $ and $r,q \in \mathbb{R}_{+}$:
    \begin{theorem}
        If $ \lim_{n\to\infty} \dfrac{f(n)}{g(n)} = 0 \Rightarrow  h \in \mathcal{O}_{q}(g) $. \\
    \end{theorem}

    \begin{corollary} 
        \[  \lim_{n\to\infty} \dfrac{f(n)}{g(n)} = 0 \Rightarrow \mathcal{O}_{r}(f) + \mathcal{O}_{q}(g) = \mathcal{O}_{q}(g)\]
    \end{corollary}

    \begin{theorem}
    \end{theorem}

    \begin{corollary}
        \[ \lim_{n\to\infty} \dfrac{f(n)}{g(n)} = \infty \Rightarrow \mathcal{O}_{r}(f) + \mathcal{O}_{q}(g) = \mathcal{O}_{r}(f)\]
    \end{corollary}

    \begin{theorem}
        If $ \lim_{n\to\infty} \dfrac{f(n)}{g(n)} = t, \ t \in \mathbb{R}_{+} \Rightarrow  h \in \mathcal{O}_{r} \left( f + \dfrac{r}{q} \cdot g \right) $. \\
    \end{theorem}
    \begin{corollary}
        \[ \lim_{n\to\infty} \dfrac{f(n)}{g(n)} = t, \ t \in \mathbb{R}_{+} \Rightarrow \mathcal{O}_{r}(f) + \mathcal{O}_{q}(g) = \mathcal{O}_{r}(f + \dfrac{r}{q} \cdot g)\]
    \end{corollary}

The proofs for the latest results is similar with the proof for addition in Big Theta, by using conversion presented for Big r-O class.

For addition in Big r-Omega, the following relations hold for any correctly defined functions $f, g, f', g', h:\mathbb{N}\longrightarrow\mathbb{R}$, where $ h(n) = f'(n) + g'(n)\  \forall n \in \mathbb{N} $, where $f',g'$ are two arbitrary functions such that $ f' \in \Omega_{r}(f), g' \in \Omega_{q}(g) $ and $r,q \in \mathbb{R}_{+}$:
    \begin{theorem}
        If $ \lim_{n\to\infty} \dfrac{f(n)}{g(n)} = 0 \Rightarrow  h \in \Omega_{q}(g) $. \\
    \end{theorem}

    \begin{corollary}
        \[  \lim_{n\to\infty} \dfrac{f(n)}{g(n)} = 0 \Rightarrow \Omega_{r}(f) + \Omega_{q}(g) = \Omega_{q}(g)\]
    \end{corollary}

    \begin{theorem}
        If $ \lim_{n\to\infty} \dfrac{f(n)}{g(n)} = \infty \Rightarrow  h \in \Omega_{r}(f) $. \\
    \end{theorem}

    \begin{corollary}
        \[ \lim_{n\to\infty} \dfrac{f(n)}{g(n)} = \infty \Rightarrow \Omega_{r}(f) + \Omega_{q}(g) = \Omega_{r}(f)\]
    \end{corollary}

    \begin{theorem}
        If $ \lim_{n\to\infty} \dfrac{f(n)}{g(n)} = t, \ t \in \mathbb{R}_{+} \Rightarrow  h \in \Omega_{r} \left( f + \dfrac{r}{q} \cdot g \right) $. \\
    \end{theorem}
    \begin{corollary}
        \[ \lim_{n\to\infty} \dfrac{f(n)}{g(n)} = t, \ t \in \mathbb{R}_{+} \Rightarrow \Omega_{r}(f) + \Omega_{q}(g) = \Omega_{r}(f + \dfrac{r}{q} \cdot g)\]
    \end{corollary}
    
The proofs for the latest results is similar with the proof for addition in Big Theta, using conversion presented for Big r-Omega class. 

The additions in Small r-O and Small r-Omega have the same properties as described in Bachmann-Landau notations.

\section{Conclusions}
\label{conclusion}

The r-Complexity~\cite{rc} model proposes a different approach to measuring the algorithm complexity, that seeks to provide a more refined awareness of algorithm performance, by not only performing asymptotic analysis, but also providing relevant estimations for finite inputs. 

Compared to traditional Bachmann-Landau notations, r-Complexity can distinguish between algorithms that would otherwise fall under the same complexity class, providing more granular insights for developers in deciding the optimal solution to use. Moreover, r-Complexity classes have a flexible nature, as functions defining classes can be easily converted, given different values of $r$, via the simple technique of changeover developed in~\cite{rc}. For this reasons, the r-Complexity framework has the potential to enhance the standard ways engineers and developers assess and optimize their algorithms. The presented applications in Section~\ref{motivation} demonstrated that r-Complexity provides a valuable framework for analysing algorithm performance beyond asymptotic complexity, enabling more accurate predictions(as observed on matrix multiplication) and deeper understanding of program behaviour (as exemplified by their applicability on building code embeddings). One of the key insights provided by this complexity model is that Strassen's algorithm, despite its theoretical asymptotic advantage, is slower than a traditional matrix multiplication, for practical input sizes.

The main contributions of this paper consists in enumerating and proofing some of the properties of calculus within r-Complexity classes, such as reflexivity, transitivity, symmetry, and projections. Additionally, the study investigated the behaviour of r-Complexity classes under addition and conversion rules between r-Complexity classes and the more traditional Bachmann-Landau complexity classes. These findings aim to help users of r-Complexity to speed up their calculus, and potentially support the wider adoption of this method of performing algorithm analysis and trade-off comparisons.


\end{document}